\documentclass[11pt]{article}
\usepackage{fullpage,amsmath,amsfonts,comment,booktabs,mathtools} 
\usepackage{amsthm,amssymb,framed}
\usepackage{times}
\usepackage[colorlinks=true,urlcolor=magenta,linkcolor=blue,citecolor=blue]{hyperref}

\newcommand{\sparagraph}[1]{\vspace{1mm}\noindent {\bf #1}}

\hypersetup{pdftitle={Min-Sum Subset Convolution and Exact Bushy Join Ordering},
pdfauthor={Mihail Stoian}}

\newtheorem{theorem}{Theorem}
\newtheorem{lemma}[theorem]{Lemma}
\newtheorem{corollary}[theorem]{Corollary}
\theoremstyle{definition}
\newtheorem{example}[theorem]{Example}

\newcommand{\Z}{\mathbb Z}
\newcommand{\poly}{\operatorname{poly}}
\newcommand{\polylog}{\operatorname{polylog}}
\newcommand{\OPT}{\operatorname{OPT}}
\newcommand{\Cout}{C_{\mathrm{out}}}
\newcommand{\disj}{\mathbin{\dot\cup}}
\newcommand{\FF}{\mathcal F}

\title{Truly Sub-$3^n$ Min-Sum Subset Convolution and Join Ordering}
\author{Mihail Stoian\\University of Technology Nuremberg}
\date{}

\begin{document}
\maketitle

\begin{abstract}
We present a deterministic reduction from min-sum subset convolution to min-plus matrix product. We show that if the min-plus product of two $D\times D$ matrices with $\beta$-bit integer entries can be computed in $D^{3-\delta}\poly(\beta,\log D)$ time for a fixed rational $0<\delta<1$, then min-sum subset convolution on an $n$-element universe can be solved in $(2+2^{-\delta})^n 2^{O(\sqrt n\log(n+1))}\poly(n,\beta)$ time. Instantiating this reduction with the recent breakthrough on subcubic min-plus matrix product by Alman and Vassilevska Williams gives a Las Vegas algorithm with expected running time $O^*(2.9987^n)$ and a deterministic algorithm with running time $O^*(2.9997^n)$, strictly breaking the longstanding $3^n$ computational barrier. Notably, these speedups translate directly to database query optimization, yielding the same expected and deterministic running-time bounds for join ordering under the $C_{\mathrm{out}}$ cost function.
\end{abstract}

\section{Introduction}

Consider two functions $f,g:2^{[n]}\to\Z\cup\{+\infty\}$. Their min-sum subset convolution is defined as:
\begin{equation}
 h(S)=(f\star g)(S)
 =\min_{A\disj B=S}\bigl(f(A)+g(B)\bigr),
 \qquad S\subseteq[n].
 \label{eq:convolution}
\end{equation}
A naive, direct evaluation inherently examines $3^n$ ordered disjoint pairs $(A,B)$.
Previous work by Bj\"orklund, Husfeldt, Kaski, and Koivisto~\cite{fsc} established a fast subset-convolution algorithm over rings, alongside an $\widetilde O(2^n M)$-time algorithm for min-sum subset convolution where integer values are strictly bounded in magnitude by $M$. Stoian~\cite{stoian-approx} later conjectured that exact min-sum subset convolution could not achieve $O((3-\varepsilon)^n\polylog M)$-time for any constant $\varepsilon>0$.

In the domain of database query optimization, this very primitive appears in the problem of join ordering, under the setting of bushy query plans, with many advances on the exact algorithm since the infancy of the field~\cite{selinger,  Cluet95onthe, vance_maier, clique_np_hard, vance_phdthesis, lohman_cross_products, dpccp, dphyp, jo_as_sp}. Stoian and Kipf~\cite{dpconv} formalized this alignment in \texttt{DPconv}, demonstrating an exact algorithm for the $\Cout$ objective that relies on the bounded-integer subset-convolution algorithm. While the running time of their algorithm for the alternative $C_{\max}$ objective avoids the dependence on the largest cardinality $W$, their exact $\Cout$ optimization exhibits a linear dependency in this parameter, namely $O^*(2^n W)$.

\sparagraph{Contribution.}~We introduce a deterministic reduction from min-sum subset convolution to min-plus matrix products. This reduction implies a truly sub-$3^n$ algorithm given the recent subcubic matrix product algorithm~\cite{breakthrough}. Crucially, our algorithm is exact, imposes no polynomial bound on the magnitude of
the integer inputs, and has only polynomial dependence on their binary
encoding length. We show that the similar time bound holds in the setting of join ordering via the connection established in \cite{dpconv}; similar improvements are natural to other problems, whose solutions use this primitive (see, e.g., \cite{fomin_exp_algos, fsc, stoian-approx} for a few examples). Needless to say, the latter result, touching a fundamental problem in query optimization, comes as a (huge) surprise to the author.


\sparagraph{Computational Model.}~Throughout, all finite values of the input functions and all
cardinality estimates are integers represented in binary.
Let $\beta\ge1$ bound their bit length.
For subset convolution, the two input functions are supplied
as explicit tables of $2^n$ values each; these tables may also
contain the symbol $+\infty$. We use a random-access model and charge arithmetic according
to operand bit length.
All computations are exact, and polynomial factors in address
length are included in the bounds.
The notation $O^*(a^n)$ suppresses factors polynomial in
$n$ and $\beta$, but never polynomial factors in the full
$2^n$-entry table size.
A Las Vegas algorithm always returns the exact answer and has
the stated expected running time.

\begin{theorem}[Min-sum subset convolution]\label{thm:transfer}
Fix a rational $0<\delta<1$. Suppose the min-plus product of two $D\times D$
integer matrices with $\beta$-bit entries can be computed in
$D^{3-\delta}\poly(\beta,\log D)$ time.
Then min-sum subset convolution can be computed in time
\begin{equation}
 (2+2^{-\delta})^n
 2^{O(\sqrt n\log(n+1))}\poly(n,\beta).
 \label{eq:main-time}
\end{equation}
A deterministic product algorithm gives a deterministic convolution
algorithm. A Las Vegas product algorithm gives a Las Vegas convolution
algorithm with the corresponding expected bound.
\end{theorem}



We prove Theorem~\ref{thm:transfer} in Sec.~\ref{sec:global}.
Combining it with the min-plus algorithms of Alman and
Vassilevska Williams~\cite{breakthrough} yields the following bounds.

\begin{corollary}\label{cor:min-sum-corr}
Min-sum subset convolution on integer inputs can be computed
by a Las Vegas algorithm in expected $O^*(2.9987^n)$ time
and by a deterministic algorithm in $O^*(2.9997^n)$ time.
\end{corollary}

\begin{corollary}\label{cor:jo-corr}
The same bounds hold for join ordering under $\Cout$,
provided each cardinality estimate $c(S)$ can be obtained in
$\poly(n,\beta)$ time.
For ordinary undirected query graphs, both bounds hold
whether cross products are permitted or forbidden.
\end{corollary}

Our Lemma~\ref{lem:AW} (Appendix \ref{sec:bits}) gives min-plus exponents $2.998$ and $2.9995$
for the Las Vegas and deterministic algorithms, respectively.
Since Theorem~\ref{thm:transfer} assumes an exponent of $3-\delta$,
we take
\[
 \delta=3-2.998=0.002
 \qquad\text{and}\qquad
 \delta=3-2.9995=0.0005,
\]
respectively. Substituting these values gives
\[
 2+2^{-0.002}=2.998614666\ldots<2.9987,
 \qquad
 2+2^{-0.0005}=2.999653486\ldots<2.9997.
\]
In each case, the gap between the two constants absorbs the
subexponential factor in Eq.~\eqref{eq:main-time}.
We justify the min-plus bounds in Sec.~\ref{sec:bits}
and prove the join-ordering result in Sec.~\ref{sec:join}.\footnote{
A prototype implementation is available at
\url{https://github.com/stoianmihail/LeDP}.}

\section{Restricted Matrix Products}\label{sec:matrix}

\sparagraph{Partitions and Restrictions.}~A candidate evaluation in Eq.~\eqref{eq:convolution} is defined by two operands, $f(A)$ and $g(B)$. Across all possible output sets $S$, every element in the universe $[n]$ must assume exactly one of three states: (i) it belongs to the left child $A$, (ii) to the right child $B$, or (iii) to the excluded outside set $C = [n] \setminus S$. Consequently, any candidate intrinsically forms an ordered ternary partition $(A,B,C)$ of the universe $[n]$. For any fixed target set $S$, the outside set $C$ remains constant while the elements of $S$ are distributed between $A$ and $B$.

To systematically divide the evaluation space, we introduce a secondary ordered partition of the universe, $[n]=X\disj Y\disj Z$, which serves as a structural restriction on the valid candidates. Rather than representing physical subplans, these sets act as indexing filters that constrain element membership according to the following strict conditions:

\begin{center}
\begin{tabular}{@{}lll@{}}
\toprule
Restriction part & Allowed candidate memberships & Forbidden membership \\
\midrule
$X$ (middle)   & \emph{left} child $A$ $\lor$ \emph{right} child $B$ & \emph{outside} set $C$ \\
$Y$ (columns)  & \emph{right} child $B$ $\lor$ \emph{outside} set $C$ & \emph{left} child $A$ \\
$Z$ (rows)     & \emph{left} child $A$ $\lor$ \emph{outside} set $C$ & \emph{right} child $B$ \\
\bottomrule
\end{tabular}
\end{center}

Formally, a candidate partition $(A,B,C)$ is \emph{compatible} with a given restriction $(X,Y,Z)$ if and only if its forbidden intersections are empty, namely:
\begin{equation}
 X\cap C=\varnothing,\qquad
 Y\cap A=\varnothing,\qquad
 Z\cap B=\varnothing.
 \label{eq:compatibility}
\end{equation}

\begin{lemma}[Restricted product]\label{lem:product}
For a restriction $(X,Y,Z)$, define matrices indexed by
$U\subseteq Z$, $T\subseteq X$, and $V\subseteq Y$:
\begin{equation}
 F[U,T]=f(U\cup T),
 \qquad
 G[T,V]=g((X\setminus T)\cup V).
 \label{eq:matrices}
\end{equation}
The entry $(F\star_{\min,+}G)[U,V]$ is the minimum of
$f(A)+g(B)$ over all compatible candidates with
$A\disj B=X\cup U\cup V$.
\end{lemma}

\begin{proof}
For every middle index $T\subseteq X$, the operands correspond to
\[
 A=U\cup T,\qquad B=(X\setminus T)\cup V.
\]
The sets are disjoint and their union is $X\cup U\cup V$.
They satisfy Eq.~\eqref{eq:compatibility}.
Conversely, any compatible candidate has the unique representation
\[
 T=A\cap X,\qquad U=A\cap Z,\qquad V=B\cap Y.
\]
Compatibility forces all of $X$ into $A\cup B$, so $B\cap X=X\setminus T$.
It also gives $A\cap Y=B\cap Z=\varnothing$.
Consequently, Eq.~\eqref{eq:matrices} enumerates precisely the compatible
candidates, and taking the minimum proves the statement.
\end{proof}

Let us see next a concrete example of this construction.

\begin{example}\label{ex:binary}
Let $[n]=\{1,2,3,4\}$, let $S=\{1,2,3\}$, and consider the binary
split $A=\{1,3\}$ and $B=\{2\}$. The outside set is $C=\{4\}$.
Choose $X=\{1,2\}$, $Y=\{4\}$, and $Z=\{3\}$.
The row $U=\{3\}$, middle index $T=\{1\}$, and column $V=\varnothing$
produce exactly
\[
 F[\{3\},\{1\}]+G[\{1\},\varnothing]
 =f(\{1,3\})+g(\{2\}).
\]
With this, we can update the result of $h(\{1,2,3\})$.
\end{example}

\begin{lemma}\label{lem:restriction-cost}
Under the assumption of Theorem~\ref{thm:transfer}, let
$m=\min\{|X|,|Y|,|Z|\}$. All updates from one restriction can be computed
in $2^{n-\delta m}\poly(n,\beta)$ time.
The same statement holds in expectation for a Las Vegas product algorithm.
\end{lemma}

\begin{proof}
The matrix dimensions in Eq.~\eqref{eq:matrices} are
$2^{|Z|}\times2^{|X|}$ and $2^{|X|}\times2^{|Y|}$.
Tile the three index ranges into blocks of length $D=2^m$.
There are $2^{n-3m}$ triples of blocks, each requiring one $D\times D$
min-plus product. Their total running time is
\[
 2^{n-3m}D^{3-\delta}\poly(n,\beta)
 =2^{n-\delta m}\poly(n,\beta).
\]
Take entrywise minima over the inner blocks.
The total number of entries read or updated over all block triples is
$O(2^{n-m})$, which is at most $O(2^{n-\delta m})$ because $\delta<1$.
Note that index generation and table access add polynomial factors in $n$. This argument also covers $m=0$, when the square products have dimension one.
\end{proof}

\sparagraph{Resolution of $+\infty$ Entries.}~The square matrix primitive requires only finite integer operands. Let $W\ge1$ bound the absolute values of all finite entries within $f$ and $g$, substituting $+\infty$ with the finite proxy $H=4W+1$. Consequently, the sum of any two valid finite entries lies in $[-2W,2W]$, whereas any sum incorporating a proxy replacement strictly exceeds $2W$. Following each matrix product, values exceeding $2W$ are deterministically mapped back to $+\infty$. This replacement mechanism requires only $O(\beta)$ bits and mathematically preserves the integrity of every valid update.
\section{A Deterministic Covering Lemma}\label{sec:cover}

Write $q=2^{-\delta}\in(1/2,1)$.
For a restriction $\sigma=(X_\sigma,Y_\sigma,Z_\sigma)$ on a block of
$r$ coordinates, set
\begin{equation}
 w(\sigma)=q^{|X_\sigma|}+q^{|Y_\sigma|}+q^{|Z_\sigma|}.
 \label{eq:weight}
\end{equation}
A family is \emph{complete} if every ordered partition $(A,B,C)$ of the
block is compatible with some family member.
Let $H_t=\sum_{j=1}^t 1/j$.

\begin{lemma}[Complete local covers]\label{lem:cover}
For every $r\ge1$, one can deterministically construct in
$2^{O(r)}$ bit operations a complete family $\FF_r$ satisfying
\begin{equation}
 \sum_{\sigma\in\FF_r}w(\sigma)
 \le 6H_{3^r}\left(\frac{2+q}{2}\right)^r.
 \label{eq:local-cover}
\end{equation}
The constants in the construction may depend on the fixed rational
$\delta$.
\end{lemma}
\begin{proof}
Consider the universe $\mathcal U_r=\{A,B,C\}^r$ of candidate
membership strings.
For each restriction $\sigma$, let
\[
 K_\sigma
 =\{u\in\mathcal U_r \mid
       u\text{ is compatible with }\sigma\}
\]
be the set of candidates it covers.

For each candidate and each coordinate, exactly two restriction
labels are compatible with its membership label.
Consequently, every candidate belongs to exactly $2^r$ sets
$K_\sigma$.

We first bound the total weight of all restrictions.
For any fixed label $L\in\{X,Y,Z\}$, summing
$q^{|L_\sigma|}$ over all restriction strings gives $(2+q)^r$:
each coordinate contributes $q$ when assigned label $L$ and
$1$ for either of the other two labels.
Therefore,
\[
 \sum_\sigma w(\sigma)=3(2+q)^r.
\]

To implement the construction with finite precision, round each
$q^k$ upward to a multiple of $2^{-r}$, for $0\le k\le r$.
Since $q^k\ge2^{-r}$, the rounded quantity lies between $q^k$
and $2q^k$.
The resulting rational weights satisfy
\[
 w(\sigma)\le\widehat w(\sigma)\le2w(\sigma),
 \qquad
 \sum_\sigma\widehat w(\sigma)\le6(2+q)^r.
\]

Run weighted greedy set cover~\cite{chevatal}.
Initially, all candidates are uncovered.
At each step, let $U_{\rm rem}$ be the set of uncovered candidates
and define
\[
 d_\sigma=|K_\sigma\cap U_{\rm rem}|.
\]
Choose a restriction minimizing
$\widehat w(\sigma)/d_\sigma$ among those with $d_\sigma>0$,
add it to $\FF_r$, and mark its covered candidates as covered.
Continue until no uncovered candidate remains.

Suppose $u=|U_{\rm rem}|>0$.
Since every remaining candidate belongs to exactly $2^r$
sets $K_\sigma$, double counting gives
\[
 \sum_\sigma d_\sigma=u\,2^r,
\]
and, by averaging,
\begin{align*}
 \min_{\sigma:\,d_\sigma>0}
 \frac{\widehat w(\sigma)}{d_\sigma}
 &\le
 \frac{\sum_\sigma\widehat w(\sigma)}
      {\sum_\sigma d_\sigma}\\
 &\le
 \frac{6}{u}\left(\frac{2+q}{2}\right)^r.
\end{align*}

Set
\[
 B_r=6\left(\frac{2+q}{2}\right)^r.
\]
If the greedy choice covers $d$ previously uncovered candidates
when $u$ remain, its rounded weight is at most $dB_r/u$.
Moreover,
\[
 \frac{d}{u}
 \le \sum_{j=u-d+1}^{u}\frac1j.
\]
Across all greedy steps, these summation ranges partition
$\{1,\ldots,3^r\}$.
Thus
\[
 \sum_{\sigma\in\FF_r}w(\sigma)
 \le
 \sum_{\sigma\in\FF_r}\widehat w(\sigma)
 \le B_rH_{3^r},
\]
which proves Eq.~\eqref{eq:local-cover}.
The family is complete because the algorithm terminates only
after covering every candidate.

Finally, there are $3^r$ restrictions and $3^r$ candidates.
Explicit incidence enumeration and a direct implementation of
greedy take $2^{O(r)}$ time.
All rounded weights have $O(r)$ bits.
For fixed rational $\delta=a/d$, the rounded numerator of $q^k$
is
\[
 \left\lceil 2^{r-ak/d}\right\rceil.
\]
It can be computed by integer root comparison, by finding the
smallest integer $t$ such that
\[
 t^d\ge2^{dr-ak}.
\]
The integers involved have $O(r)$ bits, with constants depending
on the fixed $\delta$.
Hence the entire construction takes $2^{O(r)}$ bit operations.
\end{proof}

\section{Truly Sub-$3^n$ Min-Sum Subset Convolution}
\label{sec:global}

Partition $[n]$ into $k$ disjoint coordinate blocks of sizes
$r_1,\ldots,r_k$, where
$1\le r_j\le r=\lceil\sqrt n\rceil$ and $k=\lceil n/r\rceil$.
Construct the families of Lemma~\ref{lem:cover} for these sizes.
Let $\FF$ consist of all restrictions obtained by choosing one member
from each local family and concatenating them.
The Cartesian product can be enumerated without materializing it.

For each restriction $\sigma\in\FF$, let $F_\sigma$ and $G_\sigma$
be the matrices from Lemma~\ref{lem:product}; explicitly,
\[
 F_\sigma[U,T]=f(U\cup T),
 \qquad
 G_\sigma[T,V]=g((X_\sigma\setminus T)\cup V),
\]
where $U\subseteq Z_\sigma$, $T\subseteq X_\sigma$, and
$V\subseteq Y_\sigma$.
Compute the min-plus product \[
 P_\sigma=F_\sigma\star_{\min,+}G_\sigma,
\]
which is a $2^{|Z_\sigma|}\times2^{|Y_\sigma|}$ matrix,
with rows indexed by $U\subseteq Z_\sigma$ and columns indexed
by $V\subseteq Y_\sigma$.
Its entries are
\begin{align*}
 P_\sigma[U,V]
 &=\min_{T\subseteq X_\sigma}
   \bigl(F_\sigma[U,T]+G_\sigma[T,V]\bigr)\\
 &=\min_{T\subseteq X_\sigma}
   \bigl(f(U\cup T)+g((X_\sigma\setminus T)\cup V)\bigr).
\end{align*}
For fixed $U,V$, each choice of $T$ specifies a binary split
\[
 A=U\cup T,
 \qquad
 B=(X_\sigma\setminus T)\cup V
\]
of the same output set $S=X_\sigma\cup U\cup V$.
By Lemma~\ref{lem:product}, these are precisely the splits
of $S$ compatible with $\sigma$.
Thus $P_\sigma[U,V]$ is the minimum cost among those splits
and supplies an update to $h(S)$.

Initialize every entry of $h$ to $+\infty$.
For each $\sigma\in\FF$, compute $P_\sigma$ and perform
\[
 h(X_\sigma\cup U\cup V)
 \gets
 \min\{h(X_\sigma\cup U\cup V),\,P_\sigma[U,V]\}
\]
for all its row and column indices.

\begin{proof}[Proof of Theorem~\ref{thm:transfer}]
Every matrix update uses a valid disjoint pair, by
Lemma~\ref{lem:product}.
Conversely, fix any disjoint pair $(A,B)$ and put
$C=[n]\setminus(A\cup B)$.
On each coordinate block, completeness supplies a compatible local
restriction. Their concatenation belongs to $\FF$ and is compatible
with $(A,B,C)$.
Thus every candidate is considered at least once.
Repeated candidates cause no difficulty because the minimum
operation is idempotent.
The returned table is exactly $h$, as in
Eq.~\eqref{eq:convolution}.

It remains to bound all work, including unbalanced restrictions.
Since $q<1$, Lemma~\ref{lem:restriction-cost} and
\[
 q^{\min\{|X|,|Y|,|Z|\}}
 \le q^{|X|}+q^{|Y|}+q^{|Z|}
\]
give, after suppressing a common $\poly(n,\beta)$ factor,
\begin{align}
 T
 &\le 2^n\sum_{\sigma\in\FF}
       \bigl(q^{|X_\sigma|}
             +q^{|Y_\sigma|}
             +q^{|Z_\sigma|}\bigr)
       \nonumber\\
 &=2^n\sum_{L\in\{X,Y,Z\}}
       \prod_{j=1}^k
       \left(
         \sum_{\tau\in\FF_{r_j}}q^{|L_\tau|}
       \right)
       \nonumber\\
 &\le 3\cdot2^n
       \prod_{j=1}^k
       \left[
         6H_{3^{r_j}}
         \left(\frac{2+q}{2}\right)^{r_j}
       \right]
       \nonumber\\
 &=3(2+q)^n\prod_{j=1}^k6H_{3^{r_j}}
       \nonumber\\
 &\le (2+q)^n\,2^{O(\sqrt n\log(n+1))}.
 \label{eq:global-cost}
\end{align}
The equality in the second line is the distributive identity for the
Cartesian product of local families.
It imposes no balance condition on $X,Y,Z$.
In the last line we use $H_{3^{r_j}}\le1+r_j\ln3$.

The covers take $2^{O(\sqrt n)}\poly(n)$ time to construct,
and initialization and index generation do not affect the bound.
If the matrix algorithm is Las Vegas, the same calculation gives
the expected running time by linearity of expectation.
Every matrix product is computed exactly, so the convolution
is exact as well.
\end{proof}

Note that a candidate partition where $A=[n]$ and $B=C=\varnothing$ requires
$Y=\varnothing$ in every compatible restriction.
Such extreme candidates would be omitted by a covering family
restricted to equal-sized parts.
Consequently, the weighted summation in Eq.~\eqref{eq:global-cost}
accounts for all restriction shapes, explicitly including empty parts.

\section{Application: Join Ordering}\label{sec:join}

We apply our subset-convolution algorithm to the layer-wise
dynamic program of Stoian and Kipf~\cite[Sec.~3.1]{dpconv}.
We call the resulting algorithm \texttt{LeDP}; the change from
\texttt{DPconv} is the computation of min-sum subset convolution
described in Sec.~\ref{sec:global}.

Let $V\coloneqq[n]$ index the base relations.
For every nonempty $S\subseteq V$, let $c(S)<2^\beta$ be its
nonnegative cardinality estimate.
The estimates are supplied as a table, or each $c(S)$ can be
obtained in $\poly(n,\beta)$ time.
A query plan is a full binary tree whose leaves are the relations,
each appearing once.
An internal node with leaf set $S$ contributes $c(S)$ to the
objective $\Cout$.

The well-known dynamic programming recurrence is
\begin{equation}
 \OPT(S)=
 \begin{cases}
  0, & \text{if } |S|=1,\\
  c(S)+
  \displaystyle\min_{\varnothing\ne A\subsetneq S}
  \bigl(\OPT(A)+\OPT(S\setminus A)\bigr),
  & \text{if } |S|\ge2.
 \end{cases}
 \label{eq:join-dp}
\end{equation}

\begin{proof}[Proof of Corollary~\ref{cor:jo-corr}]
Initialize $\OPT(\{i\})=0$ for every $i\in V$.
For $k=2,\ldots,n$, define
\[
 f_k(S)=
 \begin{cases}
  \OPT(S), & \text{if } 1\le |S|<k,\\
  +\infty, & \text{otherwise}.
 \end{cases}
\]
Compute $h_k=f_k\star f_k$ and set
\[
 \OPT(S)=c(S)+h_k(S)
 \qquad\text{for all } |S|=k.
\]
Every proper nonempty child of such an $S$ has size less than $k$,
and the empty set has infinite value in $f_k$.
Induction on $k$ therefore gives Eq.~\eqref{eq:join-dp}.

There are $n-1$ convolution calls.
Since a plan has $n-1$ internal nodes, every finite DP value is
at most $(n-1)2^\beta$ and has
$O(\beta+\log(n+1))$ bits.

To forbid cross products in an ordinary undirected query graph,
keep $\OPT(S)=+\infty$ whenever the induced subgraph on $S$
is disconnected, and update only connected sets.
If two connected children $A,B$ have a connected union,
an edge must cross from $A$ to $B$.
Thus every finite transition is a valid non-Cartesian join.
Conversely, every cross-product-free binary plan has connected
leaf sets at all nodes, so no such plan is excluded.
Testing connectivity for all subsets takes $2^n\poly(n)$ time.

If $\OPT(V)=+\infty$, no permitted plan exists.
Otherwise, reconstruct an optimal plan as
in~\cite[Alg.~2]{dpconv}: at each nonsingleton set $S$,
find a proper nonempty subset $A$ satisfying
\[
 \OPT(S)=c(S)+\OPT(A)+\OPT(S\setminus A),
\]
and recurse on the two children.
There are at most $n-1$ such searches, each examining at most
$2^n$ subsets, so reconstruction takes
$2^n\poly(n,\beta)$ time.

Applying Corollary~\ref{cor:min-sum-corr} to the $n-1$
convolution calls gives expected $O^*(2.9987^n)$ Las Vegas
time and $O^*(2.9997^n)$ deterministic time.
\end{proof}

\sparagraph{AI Storytime.}~We used ChatGPT-6 Astra $\{\texttt{xhigh}, \texttt{max}\}$ to develop the algorithm for min-sum subset convolution once we learned about the breakthrough in \cite{breakthrough}. We hinted the model at the split-and-list technique developed in \cite{maxcut-williams}. The author takes full responsibility for the content of this paper.

\bibliography{le-dp}
\bibliographystyle{alphaurl}

\appendix

\section{Instantiating the Min-Plus Primitive}\label{sec:bits}

The following lemma records the bit-complexity bounds obtained
from the algorithms of Alman and Vassilevska
Williams~\cite{breakthrough}.
Part~(i) follows by analyzing their real-valued algorithm
(Theorem~35) on integer inputs.
For part~(ii), their Theorem~22 gives the deterministic bound
for polynomially bounded weights; the proof below extends
its analysis to arbitrary $\beta$-bit integers.

\begin{lemma}[Bit complexity of min-plus product]\label{lem:AW}
The min-plus product of two $D\times D$ matrices with $\beta$-bit
integer entries can be computed:
\begin{enumerate}
 \item[(i)] by a Las Vegas algorithm in expected time
 $D^{2.998}\poly(\beta,\log D)$;
 \item[(ii)] by a deterministic algorithm in time
 $D^{2.9995}\poly(\beta,\log D)$.
\end{enumerate}
Both algorithms return the exact product.
\end{lemma}

\begin{proof}
For (i), use Theorem~35 of~\cite{breakthrough}.
Although stated on a real RAM, its input-dependent arithmetic
uses sums and differences of a constant number of original entries,
together with comparisons; see their Lemmas~36(a) and~37 and
the reductions of~\cite{more-believable}.
These values have $O(\beta+1)$ bits.
The auxiliary counting computations use integers of
$O(\log D)$ bits, as in their Section~4.3.
Accounting for bit arithmetic multiplies the stated expected
running time by at most $\poly(\beta,\log D)$.

For (ii), their Theorem~22 assumes polynomially bounded weights,
so we must retain the weight dependence in its proof.
Their deterministic min-plus algorithm reduces to Exact Triangle,
which asks whether an edge-weighted graph contains a triangle
of total weight zero.
We check the bit-length dependence of this intermediate algorithm.

Consider an Exact Triangle instance with $s$ vertices per part
and weights of absolute value below $2^\beta$.
In the prime-hashing argument of their Theorem~17, a nonzero
triangle sum has absolute value below $3\cdot2^\beta$.
It is therefore divisible by at most
$O((\beta+1)/\log L)$ primes in $[\sqrt L/2,\sqrt L)$,
where $L$ is their auxiliary dimension.
The same averaging argument shows that the deterministically
selected prime has at most
\[
 F(p)=O\!\left(\frac{s^3(\beta+1)}{\sqrt L}\right)
\]
false positives.
Their witness search consequently performs at most
\[
 O\!\left((F(p)+1)\frac{\sqrt L}{g}\right)
 =O\!\left(\frac{s^3(\beta+1)}{g}\right)
\]
checks.

Take $L=\Theta(s^{1/18})$ and $g=\Theta(L^{1/36})$,
as in the first construction of their Theorem~19.
Arithmetic on original weights has polynomial cost in $\beta$,
while the auxiliary counting computations use $O(\log s)$-bit
integers.
The remaining analysis is unchanged, giving
\[
 T_{\mathrm{ET}}(s,\beta)
 =s^{3-1/648}\poly(\beta,\log s).
\]
The reduction in their Theorem~21(b) then gives min-plus product
in
\[
 D^2 T_{\mathrm{ET}}(\lceil D^{1/3}\rceil,\beta+O(1))
 \poly(\beta,\log D)
 =
 D^{3-1/1944}\poly(\beta,\log D)
\]
time.
Since $3-1/1944<2.9995$, this proves (ii).
\end{proof}

\begin{proof}[Proof of Corollary~\ref{cor:min-sum-corr}]
By Lemma~\ref{lem:AW}, we can apply Theorem~\ref{thm:transfer} with \[
 \delta=3-2.998=\frac1{500}
 \qquad\text{and}\qquad
 \delta=3-2.9995=\frac1{2000}
\]
for the Las Vegas and deterministic algorithms, respectively.
The resulting bases satisfy \[
 2+2^{-1/500}<2.9987,
 \qquad
 2+2^{-1/2000}<2.9997.
\]
Since both inequalities are strict, the subexponential factor
in Eq.~\eqref{eq:main-time} is absorbed into the larger bases.
This gives expected $O^*(2.9987^n)$ Las Vegas time and
$O^*(2.9997^n)$ deterministic time.
\end{proof}

\end{document}